\documentclass[onefignum,onetabnum]{siamart171218}

\usepackage{lipsum}
\usepackage{amsfonts}
\usepackage{graphicx}
\usepackage{epstopdf}
\usepackage{algpseudocode}
\usepackage{amsmath}
\usepackage{amssymb}
\usepackage{xspace}
\usepackage{braket}
\ifpdf
  \DeclareGraphicsExtensions{.eps,.pdf,.png,.jpg}
\else
  \DeclareGraphicsExtensions{.eps}
\fi

\newsiamremark{remark}{Remark}
\newsiamremark{hypothesis}{Hypothesis}
\crefname{hypothesis}{Hypothesis}{Hypotheses}
\newsiamthm{claim}{Claim}

\headers{Fast numerical generation of Lie closure}{Y. Iiyama}

\title{Fast numerical generation of Lie closure\thanks{
\funding{This work was supported by Japan
Society for the Promotion of Science (JSPS),
through Grants-in-Aid for Scientific Research (KAKENHI)
Grant No. 20K22347., and by the Center of Innovations
for Sustainable Quantum AI (JST Grant Number JPMJPF2221).}}}

\author{Yutaro Iiyama\thanks{ICEPP, The University of Tokyo, Tokyo, Japan 
  (\email{iiyama@icepp.s.u-tokyo.ac.jp}).}}

\usepackage{amsopn}
\DeclareMathOperator{\diag}{diag}
\DeclareMathOperator{\poly}{poly}
\DeclareMathOperator{\tr}{tr}
\DeclareMathOperator{\proj}{proj}
\DeclareMathOperator{\vspan}{span}

\makeatletter
\newcommand*{\addFileDependency}[1]{
  \typeout{(#1)}
  \@addtofilelist{#1}
  \IfFileExists{#1}{}{\typeout{No file #1.}}
}
\makeatother

\newcommand{\SUd}{\ensuremath{SU(d)}\xspace}
\newcommand{\sud}{\ensuremath{\mathfrak{su}(d)}\xspace}
\newcommand{\algebra}{\ensuremath{\mathfrak{g}}\xspace}
\newcommand{\generators}{\ensuremath{\mathcal{G}}\xspace}
\newcommand{\closure}{\ensuremath{\mathfrak{c}}\xspace}
\newcommand{\parameters}{\ensuremath{\boldsymbol{\theta}}\xspace}
\newcommand{\evolution}{\ensuremath{U(\parameters)}\xspace}
\newcommand{\comm}[2]{\ensuremath{\left[{#1}, {#2}\right]}}
\newcommand{\innerprod}[2]{\ensuremath{\left\langle{#1}, {#2}\right\rangle}}

\algrenewcommand\algorithmicrequire{\textbf{Input:}}
\algrenewcommand\algorithmicensure{\textbf{Output:}}

\ifpdf
\hypersetup{
  pdftitle={Fast numerical generation of Lie closure},
  pdfauthor={Y. Iiyama}
}
\fi

\begin{document}

\maketitle

\begin{abstract}
  Finding the Lie-algebraic closure of a handful of matrices has important applications in quantum computing and quantum control. For most realistic cases, the closure cannot be determined analytically, necessitating an explicit numerical construction. The standard construction algorithm makes repeated calls to a subroutine that determines whether a matrix is linearly independent from a potentially large set of matrices. Because the common implementation of this subroutine has a high complexity, the construction of Lie closure is practically limited to trivially small matrix sizes. We present efficient alternative methods of linear independence check that simultaneously reduce the computational complexity and memory footprint. An implementation of one of the methods is validated against known results. Our new algorithms enable numerical studies of Lie closure in larger system sizes than was previously possible.
\end{abstract}



\section{Introduction}

In the field of quantum computing and quantum optimal control, identification of the closure of a set of elements of a finite-dimensional Lie algebra has important applications. Particularly in quantum computing, the dimension of the so-called dynamical Lie algebra (DLA) of a parametrized quantum circuit, which is the closure of the generators of parametrized quantum gates, is closely related to its trainability during variational optimization tasks~\cite{Larocca2022-bh, Fontana2024-le, Monbroussou2025-he}. Ability to promptly compute the DLA starting from a numerical representation of generators for different system sizes helps in predicting the utility of various quantum circuit structures, and to further deepen our understanding of the relationship of the algebraic, geometric, and computational properties of quantum circuits. Even apart from such fields of study, numerically studying the structure of Lie algebras can be generally interesting.

Lie closure, i.e., the smallest Lie subalgebra that contains the given elements, can be determined analytically for some special cases~\cite{Oszmaniec2017-cf, Larocca2022-bh, Pozzoli2022-oz, Aguilar2024-rb, Wiersema2024-om, Kokcu2024-vn}. However, for a majority of the nontrivial cases, numerical or symbolic but explicit construction appears to be necessary. The standard and thus far the only known strategy for construction of Lie closures is to iteratively build a basis of the subalgebra by calculating all possible nested Lie brackets and checking the linear independence of each with respect to other elements of the basis~\cite{Schirmer2001-fj, Zeier2010-me, Larocca2022-bh}. The complexity of an algorithm following this strategy for a subalgebra of dimension $N$ is in general $\mathcal{O}(N^p)$ with $p \geq 2$, where the lower bound for $p$ arises from the need to evaluate at least $N(N-1)/2$ Lie brackets. Since in many cases $N = \mathcal{O}(c^n)$ for some $c>1$ and the system size (for example, the number of qubits in quantum computing) $n$, this approach is intractable beyond some moderate $n$. In fact, it may become practically infeasible even at a small $n$, as $p$ can be much greater than 2 if the linear independence check is performed through a poorly scaling method, such as by investigating the rank of the matrix composed of the elements of the basis represented as column vectors.
An efficient procedure to determine the linear independence is therefore highly desired.


In this paper, we introduce two alternative modifications to the closure construction algorithm that employ different methods of linear independence check. Both methods not only have lower complexity than the reference algorithm based on matrix rank calculation, but also have significant advantages in handling inputs in a sparse representation.

This paper is organized as follows. In \cref{sec:background}, we introduce the scientific background of the problem of DLA construction. We then describe the general algorithm and present our improvements in \cref{sec:main}. We validate our implementations of the proposed methods in \cref{sec:validation}, and conclude in \cref{sec:conclusion}.

\section{Background}
\label{sec:background}

\subsection{Dynamical lie algebra}
\label{subsec:dla}

In quantum computation under the quantum gate model, a set of $n$ qubits is prepared in some initial state, and is transformed into a final state via an application of a quantum circuit (a sequence of quantum gates). In this paper, we focus our attention to parametrized quantum circuits (PQC), where the gates are elements of continuous families of operations that include the identity operation. In a class of applications of PQC called variational quantum algorithms (VQA), the parameters of the gates are optimized to realize the target final state. Likewise, in quantum optimal control, the main objective is to effect a transformation of a quantum system from a given initial state to a desired final state.

In all cases, the quantum states are represented by unit-norm vectors in a $d$-dimensional Hilbert space for some finite $d$ ($d=2^n$ for a $n$-qubit system). The transformations are thus represented by elements of \SUd, and its properties are conveniently analyzed using the corresponding Lie algebra \sud. Specifically, possible forms of control that can be applied to the system (parametrized gates in computing; pulse signals in optimal control) are given in terms of a small set $\generators \subset \sud$, and the full state transformation is specified by a set of parameters $\parameters = \{\theta_j\}_{j=0}^{J-1} \in \mathbb{R}^{J}$ as
\begin{equation}
    \evolution = \prod_{j=0}^{J-1} \exp (\theta_j h_{l_j}) \quad (h_l \in \generators).
\end{equation}
The closure \closure of \generators, i.e., the smallest subalgebra of \sud that contains all elements of \generators, therefore dictates which final states can be reached, at least in principle, from a given initial state. In the context of quantum optimal control, this subalgebra is called the dynamical Lie algebra (DLA) of \generators, a term later adopted in quantum computing. A $d$-dimensional system whose DLA is isomorphic to \sud, and therefore has the maximal dimension, is said to be controllable, since arbitrary final states can be obtained via \evolution for a sufficiently large number of parameters $J$.

The significance of the DLA goes beyond determining the controllability of the system. Indeed, it was recently revealed that the variance of the gradients of the cost function of VQA, which indicates how easily the optimal parameter values can be found, scales inversely with the dimension of the DLA~\cite{Larocca2022-bh,Fontana2024-le}. In other words, circuits with large DLA dimensions have vanishing cost function gradients and can become practically untrainable. At the same time, having more circuit parameters than the dimension of the DLA was shown to be a sufficient condition to guarantee convergence of the VQA to the global optimum~\cite{Larocca2023-qy}. It was even shown that the DLA of a quantum circuit can be used to classically simulate the circuit~\cite{Goh2023-ya}. Through these results, the study of the DLA for any given PQC structure has become a major tool in understanding its performance in a VQA setting.

\subsection{Nomenclature and notation}
\label{subsec:notation}

In this paper, we restrict our discussion to finite-dimensional Lie algebras over $\mathbb{C}$, and thus equate their elements with $d \times d$ complex square matrices. The presented algorithms are applicable to any such matrix Lie algebras. However, with applications to quantum physics in mind, a particular attention will be given to \sud. We will also be adopting the quantum physics terminology and will call the elements of Lie algebras generically as operators.

Because we assume complex matrix Lie algebras, the Lie bracket $\comm{\cdot}{\cdot}: \algebra \times \algebra \to \algebra$ is simply the matrix commutator
\begin{equation}
    \comm{a}{b} = ab - ba \quad (a, b \in \algebra).
\end{equation}
\algebra as a linear space is also equipped with an inner product $\innerprod{\cdot}{\cdot}: \algebra \times \algebra \to \mathbb{C}$ defined as
\begin{equation}
    \innerprod{a}{b} = \frac{1}{d}\tr (a^{\dagger} b),
\end{equation}
where $a^{\dagger}$ is the conjugate transpose of $a$. The inner product with self is the operator norm:
\begin{equation}
    \lVert a \rVert = \sqrt{\innerprod{a}{a}} \geq 0,
\end{equation}
where the equality holds if and only if $a$ is the null matrix.

Elements of an $L$-tuple $A \in S^L$ for any set $S$ are addressed by $A[l] \; (l=0,\dots,L-1)$. The number of elements in a tuple is given by $|A| = L$. Concatenation of tuples $A$ and $B \in S^{L'}$ is
\begin{equation}
A + B = (A[0], \dots, A[L-1], B[0], \dots, B[L'-1]).
\end{equation}
A tuple may be organized as a column vector, in which case the notation $\vec{A} = \left(A[0], \dots, A[L-1]\right)^T$ is used. The dot product $S^L \times T^L \to R$ between vectors $\vec{A}$ and $\vec{C} \in T^L$ represents the sum over element-wise multiplications
\begin{equation}
    \vec{A} \cdot \vec{C} = \sum_{l=0}^{L-1} A[l] C[l]
\end{equation}
when the product $S \times T \to R$ is defined and $R$ is a linear space.

\section{Main results}
\label{sec:main}


The standard algorithm for constructing the DLA \closure of generators \generators is given in \cref{alg:dla_standard}.

\begin{algorithm}
\caption{Standard algorithm for DLA construction. Reformulation of Algorithm 1 in \cite{Larocca2022-bh}.}
\label{alg:dla_standard}
\begin{algorithmic}
\Require A tuple of operators $\generators \subset \algebra$
\Ensure Basis $B$ of closure \closure of \generators
\State $B \gets ()$, 0-tuple
\ForAll{$g \in \generators$}
    \If{$B + (g)$ is linearly independent}
        \State $B \gets B + (g)$
    \EndIf
\EndFor
\State $l \gets 1,\; r \gets 0$
\While{$l < |B|$}
    \For{$m = 0, \dots, r$}
        \State $h \gets \comm{B[l]}{B[m]}$
        \If{$B + (h)$ is linearly independent}
            \State $B \gets B + (h)$
        \EndIf
    \EndFor
    \State $r \gets r + 1$
    \If{$r = l$}
        \State $l \gets l + 1, \; r \gets 0$
    \EndIf
\EndWhile
\State \Return B
\end{algorithmic}
\end{algorithm}

We first remark on the memory consumption of DLA construction, where at least $|B|$ operators must be kept in memory at the end. Recalling that the operators are $d \times d$ matrices, and for $\algebra = \sud$ on a controllable system $|B| = d^2-1$, a simple matrix representation of the operators can require $d^4-d^2$ complex numbers to be stored, which can be prohibitive for $d$ that is some exponential of the physical system size.

However, it is often the case that the elements of \generators are highly sparse matrices, because for example each element represents an operation on a small subsystem, such as a pair of neighboring qubits. In such cases, sparse representations such as a sum of Pauli strings (matrix decomposition using Kronecker products of $2 \times 2$ Pauli matrices as the basis) can be employed, reducing the memory footprint of individual operators typically to $\poly(\log d)$. In fact, it is virtually always necessary to rely on a sparse representation for any nontrivial system sizes.

Looking into the runtime, \cref{alg:dla_standard} requires $M = |B| (|B|-1) / 2$ commutator calculations and $M + |\generators| - 1$ linear independence checks. 
The commutators would be evaluated according to the numerical representation of the operators. For the direct matrix representation, they would be nothing more than a difference between two matrix products. For sparse representations such as Pauli sums, algebraic relations may be exploited. In any case, there is not so much to optimize in this part of the algorithm.

On the other hand, there are multiple ways to check whether an operator is linearly independent from a set of operators, and some scale worse than others in terms of both runtime and memory consumption. For example, a typically employed method is to calculate the rank of the matrix $(B[0] \; B[1] \; \cdots \; B[N'-1] \; h)$, where the operators are expressed as column vectors by concatenating all columns of the matrix. This method is however suboptimal, because the rank of a large ($d^2 \times N'$) matrix must be calculated. The complexity of e.g. singular value decomposition (SVD) to reveal the rank is $\mathcal{O}(d^2 (N')^2)$. Furthermore, this method is incompatible with representing the operators sparsely, unless a special algorithm for calculating the rank of a set of Pauli sums, for example, exists. Without such an algorithm, the operators must be expanded into full matrices first, nullifying the advantage in memory footprint.

There are examples of alternative approaches to linear independence check. Ref. \cite{Wiersema2024-om} analyzed the DLA generated from single Pauli strings. The authors then exploited the fact that the commutator of two Pauli strings is also a Pauli string. Since distinct strings are always linearly independent, linear independence check in this case is reduced to a test of whether the new Pauli string is contained in the basis $B$. The implementation in the PennyLane library~\cite{Bergholm2020-fc} extends this idea to general sums of Pauli strings by first checking whether the new operator contains a new Pauli string as a term. When affirmative, the operator is added to the basis without further computation. Otherwise, the implementation falls back to the aforementioned rank calculation method.

Here, we propose two methods of linear independence check that is more efficient than rank calculation and also do not rely on the specifics of the operator representation. The first, called the matrix inversion method, is more memory-efficient than the second method, but may be more prone to numerical instabilities. The second method, called the orthonormalization method, is fast and numerically more stable.

\subsection{Matrix inversion method}

The basic insight of the matrix inversion method comes from the following lemma and theorem.

\begin{lemma}\label{lma:ip_matrix_invertible}
    A matrix $A$ of inner products between non-null, linearly independent operators $g_0, \dots, g_{N-1}$ where
    \begin{equation}
        A_{ij} = \innerprod{g_i}{g_j}
    \end{equation}
    is invertible.
\end{lemma}
\begin{proof}
    By the definition of the inner product $\innerprod{\cdot}{\cdot}$, $A$ is Hermitian:
    \begin{equation}
        (A_{ij})^{*} = \innerprod{g_i}{g_j}^{*} = \innerprod{g_j}{g_i} = A_{ji}.
    \end{equation}
    Therefore there exists a unitary matrix $U$ that diagonalizes $A$:
    \begin{equation}
        U^{\dagger} A U = \diag(\lambda_0, \dots, \lambda_{N-1}).
    \end{equation}
    The eigenvalues $\{\lambda_j\}_{j=0}^{N-1}$ of $A$ are then
    \begin{equation}
    \begin{split}
        \lambda_j & = \sum_{k,l} (U_{kj})^{*} A_{kl} U_{lj} \\
        & = \innerprod{\sum_{k} U_{kj} g_k}{\sum_{l} U_{lj} g_l} \\
        & = \left\lVert \sum_{k} U_{kj} g_k \right\rVert^2.
    \end{split}
    \end{equation}
    Because the operators $g_k$ are all non-null and linearly independent, their linear composition cannot result in a null operator:
    \begin{equation}
    \left\lVert \sum_{k} U_{kj} g_k \right\rVert > 0 \;\; \forall j.
    \end{equation}
    Thus, $A$ is a full-rank Hermitian matrix, which is invertible.
\end{proof}

\begin{theorem}[Efficient determination of linear independence]\label{thm:matrix_inversion}
    Linear independence of an operator $h$ with respect to a linearly independent tuple of operators $B$ can be determined if four following procedures can be performed on the operators:
    \begin{itemize}
        \item Scalar multiplication.
        \item Addition.
        \item Inner product.
        \item Comparison with the null operator.
    \end{itemize}
\end{theorem}
\begin{proof}\label{prf:matrix_inversion}
    Assume that $h$ is linearly dependent on $B$. Let the size of $B$ be $|B| = N'$. Then there exists a vector of coefficients $\vec{x} \in \mathbb{C}^{N'}$ that satisfies the following equation:
    \begin{equation}
        \vec{B} \cdot \vec{x} = h.
    \end{equation}
    Taking the inner product with $B[l]$ on both sides, we have
    \begin{equation}
        \sum_{m=0}^{N'-1} \innerprod{B[l]}{B[m]} x_m = \innerprod{B[l]}{h}.
    \end{equation}
    Because $B$ is a tuple of linearly independent operators, the matrix $A_{lm} := \innerprod{B[l]}{B[m]}$ is invertible by \cref{lma:ip_matrix_invertible}. Stacking this equation vertically for $l=0,\dots,N'-1$ and denoting $\vec{\beta} = \left(\innerprod{B[0]}{h}, \dots, \innerprod{B[N'-1]}{h}\right)^T$, we find that
    \begin{equation}
        \vec{x} = A^{-1} \vec{\beta},
    \end{equation}
    and therefore
    \begin{equation}\label{eqn:orthogonal_component}
        h - \vec{B} \cdot \left(A^{-1} \vec{\beta}\right) = 0.
    \end{equation}
    Since the equality holds iff $h$ is linearly dependent on $B$, a comparison of the left hand side against the null operator serves as a check of linear independence.
\end{proof}

The matrix inversion method utilizes the linear independence check method outlined in the proof. The DLA construction algorithm using this method is a slight modification of \cref{alg:dla_standard}, given in \cref{alg:dla_matrix_inversion}.

\begin{algorithm}
\caption{Algorithm for DLA construction using the matrix inversion method for linear independence check.}
\label{alg:dla_matrix_inversion}
\begin{algorithmic}
\Require A tuple of operators $\generators \subset \algebra$
\Ensure Basis $B$ of closure \closure of \generators
\State $B \gets ()$, 0-tuple
\State $A \gets ()$, $0 \times 0$ matrix
\ForAll{$g \in \generators$}
    \State $\bar{g} \gets g / \lVert g \rVert$
    \State $\beta \gets (\innerprod{B[0]}{\bar{g}}, \dots, \innerprod{B[|B|-1]}{\bar{g}})$
    \If{$\bar{g} - \vec{B} \cdot \left(A^{-1} \vec{\beta}\right) \neq 0$}
        \State $B \gets B + (\bar{g})$
        \State Expand $A$:
        \State $\begin{aligned} \quad & A_{i,|B|-1} = A_{|B|-1,i}^* = \beta[i]\; \text{for}\; i=0,\dots,|B|-2 \\
                             \quad & A_{|B|-1,|B|-1}=1. \end{aligned}$
        \State Compute $A^{-1}$ and store in memory
    \EndIf
\EndFor
\State $l \gets 1,\; r \gets 0$
\While{$l < |B|$}
    \For{$m = 0, \dots, r$}
        \State $h \gets \comm{B[l]}{B[m]}/\lVert \comm{B[l]}{B[m]} \rVert$
        \State $\beta \gets (\innerprod{B[0]}{h}, \dots, \innerprod{B[|B|-1]}{h})$
        \If{$h - \vec{B} \cdot \left(A^{-1} \vec{\beta}\right) \neq 0$}
            \State $B \gets B + (h)$
            \State Expand $A$:
            \State $\begin{aligned} \quad & A_{i,|B|-1} = A_{|B|-1,i}^* = \beta[i]\; \text{for}\; i=0,\dots,|B|-2 \\
                             \quad & A_{|B|-1,|B|-1}=1. \end{aligned}$
            \State Compute $A^{-1}$ and store in memory
        \EndIf
    \EndFor
    \State $r \gets r + 1$
    \If{$r = l$}
        \State $l \gets l + 1, \; r \gets 0$
    \EndIf
\EndWhile
\State \Return B
\end{algorithmic}
\end{algorithm}

For sparsely represented operators, \cref{alg:dla_matrix_inversion} does not require expansions into dense matrices, since the procedures listed in \cref{thm:matrix_inversion} can usually be performed directly on objects in such representations. Additionally, the expensive inversion of $A$ is performed only when a new element is appended to the basis, which is in contrast to the requirement in the rank calculation method to run the SVD or some other matrix decomposition subroutine for every invocation of linear independence check. In terms of memory footprint, the size of matrix $A$ is $|B| \times |B|$ at the end of the algorithm, which is comparable to $d^2 \times |B|$ of the rank calculation method, if $\algebra = \sud$ and the system is controllable. For non-controllable systems,
$|B| \times |B|$ can be significantly smaller than $d^2 \times |B|$.

All operators are unit-normalized in \cref{alg:dla_matrix_inversion}. While this is mathematically an unnecessary procedure, it was observed to be critically important for the numerical stability of the method. Without such normalization, inner products of nested commutators can in principle grow exponentially throughout the algorithm iteration, making $A$ extremely ill-conditioned.

Even with the normalization, however, $A$ can still be ill-conditioned when $|B|$ is large, leading to observable errors for certain cases in our numerical experiments. Matrix inversion is also a computationally expensive procedure, even if it is performed for a limited number of times. The orthonormalization method described in the following is numerically more robust and relies on less expensive operations, at the price of potentially increased memory footprint.

\subsection{Orthonormalization method}

The orthonormalization method for linear independence check is based on the following corollary of \cref{thm:matrix_inversion}.

\begin{corollary}
When $B$ is a tuple of orthonormal operators, the only inner products required to check the linear independence of $h$ are between $h$ and the elements of $B$.
\end{corollary}
\begin{proof}
    This follows trivially from the proof of \cref{thm:matrix_inversion}, using $\innerprod{B[i]}{B[j]} = \delta_{ij}$ in addition. Because $A$ in this case is a $|B| \times |B|$ identity matrix $I_{|B|}$, $\vec{x} = \vec{\beta}$ is the only object requiring inner product computations.
\end{proof}

For an orthonormal $B$, the dot product
\begin{equation}
    \begin{split}
    \vec{B} \cdot (A^{-1} \vec{\beta}) & = \vec{B} \cdot \vec{\beta} \\
    & = \sum_{l=0}^{|B|-1} \innerprod{B[l]}{h} B[l]
    \end{split}
\end{equation}
in \cref{eqn:orthogonal_component} is in fact the projection of $h$ onto $\vspan(B)$, denoted as $\proj_B h$. The meaning of the linear independence condition $h_{\perp} := h - \proj_B h \neq 0$ becomes clear in this case as probing the existence of an orthogonal component of $h$ with respect to $\vspan(B)$.

The DLA construction through the orthonormalization method, given in \cref{alg:dla_orthonormalization}, works by iteratively expanding two bases of the DLA. One basis will be the output that contains the linearly independent elements of \generators and the nested commutators thereof, and the other is the auxiliary orthonormalized basis used for linear independence check. Therefore, this method requires holding twice as many operators in memory than the standard algorithm, which can be a substantial liability if the operators are represented by full matrices. On the other hand, no matrix decomposition or inversion is necessary in this method, and the complexity of linear independence check is in fact $\mathcal{O}(|B|)$.

\begin{algorithm}
\caption{Modified algorithm for DLA construction that uses the orthonormalization method for linear independence check. In the algorithm, $\proj_V a$ corresponds to the projection of operator $a$ onto the span of $V$ (equivalently span of $B$), which is easily calculable since $V$ constitutes an orthonormal basis of the linear space.}
\label{alg:dla_orthonormalization}
\begin{algorithmic}
\Require A set of operators $\generators \subset \algebra$
\Ensure Basis $B$ of closure \closure of \generators
\State $B \gets ()$, 0-tuple
\State $V \gets ()$, 0-tuple
\ForAll{$g \in \generators$}
    \State $g_{\perp} \gets g - \proj_V g$
    \If{$g_{\perp} \neq 0$}
        \State $B \gets B + (g)$
        \State $V \gets V + (g_{\perp}/\lVert g_{\perp} \rVert)$
    \EndIf
\EndFor
\State $l \gets 1,\; r \gets 0$
\While{$l < |B|$}
    \For{$m = 0, \dots, r$}
        \State $h \gets \comm{B[l]}{B[m]}$
        \State $h_{\perp} \gets h - \proj_V h$
        \If{$h_{\perp} \neq 0$}
            \State $B \gets B + (h)$
            \State $V \gets V + (h_{\perp}/\lVert h_{\perp} \rVert)$
        \EndIf
    \EndFor
    \State $r \gets r + 1$
    \If{$r = l$}
        \State $l \gets l + 1, \; r \gets 0$
    \EndIf
\EndWhile
\State \Return B
\end{algorithmic}
\end{algorithm}

Memory consumption can actually be reduced if one does not require the elements of \generators and their nested commutators to be parts of the basis. Under this scenario, $B$ can be discarded altogether, and the modified algorithm would be given as \cref{alg:dla_orthonormalization_dimonly}. One nontrivial difference between \cref{alg:dla_orthonormalization} and \cref{alg:dla_orthonormalization_dimonly} is whether commutators of elements of $B$ or $V$ are calculated. That the two result in the same DLA can be proven easily.

\begin{algorithm}
\caption{Simplified DLA construction algorithm that does not preserve the original generator elements in the output basis.}
\label{alg:dla_orthonormalization_dimonly}
\begin{algorithmic}
\Require A set of operators $\generators \subset \algebra$
\Ensure Orthonormal basis $V$ of closure \closure of \generators
\State $V \gets ()$, 0-tuple
\ForAll{$g \in \generators$}
    \State $g_{\perp} \gets g - \proj_V g$
    \If{$g_{\perp} \neq 0$}
        \State $V \gets V + (g_{\perp}/\lVert g_{\perp} \rVert)$
    \EndIf
\EndFor
\State $l \gets 1,\; r \gets 0$
\While{$l < |V|$}
    \For{$m = 0, \dots, r$}
        \State $h \gets \comm{V[l]}{V[m]}$
        \State $h_{\perp} \gets h - \proj_V h$
        \If{$h_{\perp} \neq 0$}
            \State $V \gets V + (h_{\perp}/\lVert h_{\perp} \rVert)$
        \EndIf
    \EndFor
    \State $r \gets r + 1$
    \If{$r = l$}
        \State $l \gets l + 1, \; r \gets 0$
    \EndIf
\EndWhile
\State \Return V
\end{algorithmic}
\end{algorithm}


\section{Numerical experiments}
\label{sec:validation}

\subsection{Validation with known results}

We checked the validity of the orthonormalization method and our implementation of it by computing the dimensions of the DLAs of quantum circuits introduced in Ref. \cite{Larocca2022-bh}. The DLA dimensions of these circuits are either analytically derived or have been calculated by the authors using the standard algorithm. We also compare the runtime of the orthonormalization method with a reference implementation of \cref{alg:dla_standard} using the same code base and hardware, but checking for linear independence by calculating the matrix rank via SVD. Direct matrix representation of the operators is employed in the experiments. Matrix operations including SVD are performed on a graphic processing unit (GPU).

\Cref{tab:validation} lists the names of the circuits, expected DLA dimensions, and calculation results. We find that the calculated DLA dimension matches the expectation in all instances except for $n=8$ and $10$ of (d), for which the discrepancies may come from rounding-off errors during computation. We also observe that the runtime for the orthonormalization and rank calculation methods scale significantly differently.

\begin{table}
\caption{Validation of the orthonormalization method using quantum circuits in Ref. \cite{Larocca2022-bh}. Section numbers indicate where the generator definitions can be found. DLA dimensions are calculated by our implementation of \cref{alg:dla_orthonormalization_dimonly}. Runtime is in seconds. Hyphens indicate that the experiment is not attempted because the predicted runtime is too long.}
\label{tab:validation}

\begin{flushleft}
(a) Hardware-efficient ansatz; Section 4.1. DLA dimension for an $n$-qubit circuit is $4^n - 1$ (controllable system).
\end{flushleft}
\begin{center}
\begin{tabular}{l|ccccc}
Number of qubits & 2 & 3 & 4 & 5 & 6 \\
\hline
DLA dimension & 15 & 63 & 255 & 1023 & 4095 \\
Runtime (orthonormalization) & 0.5 & 1 & 2 & 11 & 330 \\
Runtime (rank calculation) & 0.6 & 3 & 251 & - & -
\end{tabular}
\end{center}
\vspace{0.5em}
\begin{flushleft}
(b) Spin glass Hamiltonian variational ansatz (HVA); Section 4.1. DLA dimension for an $n$-qubit circuit is $4^n - 1$ (controllable system).
\end{flushleft}
\begin{center}
\begin{tabular}{l|ccccc}
Number of qubits & 3 & 4 & 5 & 6 \\
\hline
DLA dimension & 63 & 255 & 1023 & 4095 \\
Runtime (orthonormalization) & 0.5 & 0.6 & 1 & 8 \\
Runtime (rank calculation) & 1 & 14 & 571 & -
\end{tabular}
\end{center}
\vspace{0.5em}
\begin{flushleft}
(c) Heisenberg XXZ model HVA without the control generator; Section 4.2. Results for zero-magnetization subspace. DLA dimension for an $n$-qubit circuit is $d_{n/2}^2 - 1$, where $d_{n/2}$ is the dimension of the Hilbert subspace.
\end{flushleft}
\begin{center}
\begin{tabular}{l|ccccc}
Number of qubits & 4 & 6 & 8 & 10 \\
\hline
$d_{n/2}$ & 4 & 10 & 38 & 126 \\
DLA dimension & 15 & 99 & 1443 & 15875 \\
Runtime (orthonormalization) & 0.7 & 2 & 5 & 1186 \\
Runtime (rank calculation) & 0.5 & 3 & 505 & -
\end{tabular}
\end{center}
\vspace{0.5em}
\begin{flushleft}
(d) Transverse field Ising model HVA with the open boundary condition; Section 5.2.2. DLA dimension for an $n$-qubit circuit is $n^2 - 1$.
\end{flushleft}
\begin{center}
\begin{tabular}{l|ccccc}
Number of qubits & 4 & 6 & 8 & 10 \\
\hline
DLA dimension & 15 & 35 & 64 & 100 \\
Runtime (orthonormalization) & 1.4 & 1.8 & 3.2 & 35 \\
Runtime (rank calculation) & 0.5 & 3 & 505 & 1390
\end{tabular}
\end{center}
\end{table}

\subsection{DLA calculation for larger systems}

To showcase the advantage of retaining sparse representations of the operators, we extend the calculation in \cref{tab:validation} (a) to larger qubit numbers, using sparse Pauli string sums as the underlying representation of the operators. Since GPUs generally struggle with sparse representations, calculations here are performed on a many-core CPU with parallelized commutator evaluations. The DLA dimension and runtime of the algorithm are given in \cref{tab:hea}, together with the maximum size of the basis matrix $d^2 \times |B|$ that would be required for the rank calculation. 

\begin{table}
\caption{Calculation of the DLA dimension of the hardware-efficient ansatz in Ref. \cite{Larocca2022-bh} using a sparse Pauli representation of the operators. Runtime is in seconds. The orthonormalization method proposed in this paper preserves sparse representations of matrices, allowing direct calculation of DLAs up to large Hilbert space dimensions.}
\label{tab:hea}
\begin{center}
\begin{tabular}{l|cccc}
Number of qubits & 6 & 7 \\ 
\hline
DLA dimension & 4095 & 16383 \\ 
Basis matrix size & $1.67 \times 10^7$ & $2.68 \times 10^8$ \\ 
Runtime & 429 & 2345 \\ 
\end{tabular}
\end{center}
\end{table}

\section{Conclusion}
\label{sec:conclusion}

In this paper, we proposed two algorithms for numerically identifying closures of Lie algebra elements. Both algorithms have the same basic flow as the standard algorithm, but employ efficient methods to check the linear independence of a matrix with respect to the growing list of basis matrices. Compared to the standard algorithm utilizing a calculation of the rank of a large matrix, the proposed algorithms have lower complexity and consumes less computer memory, allowing direct calculation of Lie closures for larger matrix sizes than was previously possible.

With the main motivation of the development being in the studies of parametric quantum circuits, for which the Lie closure is referred to as the dynamical Lie algebra (DLA), we demonstrated one of the proposed algorithms on problems of counting the dimensions of the DLAs of various quantum circuit structures. Our implementation of the algorithm is shown to reproduce the expected results for most of the experiments, with a clear scaling difference in runtime compared to the reference implementation based on matrix rank calculation.

\section*{Acknowledgments}
The author is grateful to Lento Nagano for introducing the subject and seeding the investigations that led to the presented results. This work was partially performed within the research collaboration framework between the International Center for Elementary Particles at the University of Tokyo and Toppan, inc. set up under the Quantum Innovation Initiative.

\bibliographystyle{siamplain}
\bibliography{main}

\begin{thebibliography}{10}

\bibitem{Aguilar2024-rb}
{\sc G.~Aguilar, S.~Cichy, J.~Eisert, and L.~Bittel}, {\em Full classification of pauli lie algebras}, arXiv:2408.00081,  (2024).

\bibitem{Bergholm2020-fc}
{\sc V.~Bergholm, J.~Izaac, M.~Schuld, C.~Gogolin, M.~S. Alam, S.~Ahmed, J.~M. Arrazola, C.~Blank, A.~Delgado, S.~Jahangiri, K.~McKiernan, J.~J. Meyer, Z.~Niu, A.~Száva, and N.~Killoran}, {\em {PennyLane}: Automatic differentiation of hybrid quantum-classical computations}, arXiv:1811.04968,  (2020).

\bibitem{Fontana2024-le}
{\sc E.~Fontana, D.~Herman, S.~Chakrabarti, N.~Kumar, R.~Yalovetzky, J.~Heredge, S.~H. Sureshbabu, and M.~Pistoia}, {\em Characterizing barren plateaus in quantum ansätze with the adjoint representation}, Nat. Commun., 15 (2024), p.~7171.

\bibitem{Goh2023-ya}
{\sc M.~L. Goh, M.~Larocca, L.~Cincio, M.~Cerezo, and F.~Sauvage}, {\em Lie-algebraic classical simulations for quantum computing}, arXiv:2308.01432,  (2023).

\bibitem{Kokcu2024-vn}
{\sc E.~Kökcü, R.~Wiersema, A.~F. Kemper, and B.~N. Bakalov}, {\em Classification of dynamical lie algebras generated by spin interactions on undirected graphs}, arXiv:2409.19797,  (2024).

\bibitem{Larocca2022-bh}
{\sc M.~Larocca, P.~Czarnik, K.~Sharma, G.~Muraleedharan, P.~J. Coles, and M.~Cerezo}, {\em Diagnosing barren plateaus with tools from quantum optimal control}, Quantum, 6 (2022), p.~824.

\bibitem{Larocca2023-qy}
{\sc M.~Larocca, N.~Ju, D.~García-Martín, P.~J. Coles, and M.~Cerezo}, {\em Theory of overparametrization in quantum neural networks}, Nat. Comput. Sci., 3 (2023), pp.~542--551.

\bibitem{Monbroussou2025-he}
{\sc L.~Monbroussou, E.~Z. Mamon, J.~Landman, A.~B. Grilo, R.~Kukla, and E.~Kashefi}, {\em Trainability and expressivity of hamming-weight preserving quantum circuits for machine learning}, Quantum, 9 (2025), p.~1745.

\bibitem{Oszmaniec2017-cf}
{\sc M.~Oszmaniec and Z.~Zimborás}, {\em Universal extensions of restricted classes of quantum operations}, Phys. Rev. Lett., 119 (2017), p.~220502.

\bibitem{Pozzoli2022-oz}
{\sc E.~Pozzoli, M.~Leibscher, M.~Sigalotti, U.~Boscain, and C.~Koch}, {\em Lie algebra for rotational subsystems of a driven asymmetric top}, J. Phys. A Math. Theor., 55 (2022), p.~215301.

\bibitem{Schirmer2001-fj}
{\sc S.~G. Schirmer, H.~Fu, and A.~I. Solomon}, {\em Complete controllability of quantum systems}, Phys. Rev. A, 63 (2001), p.~063410.

\bibitem{Wiersema2024-om}
{\sc R.~Wiersema, E.~Kökcü, A.~F. Kemper, and B.~N. Bakalov}, {\em Classification of dynamical lie algebras of 2-local spin systems on linear, circular and fully connected topologies}, Npj Quantum Inf., 10 (2024), p.~110.

\bibitem{Zeier2010-me}
{\sc R.~Zeier and T.~Schulte-Herbrueggen}, {\em Symmetry principles in quantum systems theory}, arXiv:1012.5256,  (2010).

\end{thebibliography}
\end{document}